\documentclass[letterpaper,UKenlish,cleveref, autoref]{lipics-v2021}

\newif\ifanon
\anonfalse

\hideLIPIcs 
\title {An Approximation Notion between {\sf P} and {\sf FPTAS}}

\author
{Samuel Bismuth}
{Department of Computer Science, Ariel University, Ariel 40700, Israel}
{samuelbismuth101@gmail.com}
{https://orcid.org/0000-0003-3471-5402}
{Israel Science Foundation grant no. 712/20.}

\author
{Erel Segal-Halevi}
{Department of Computer Science, Ariel University, Ariel 40700, Israel}
{erelsgl@gmail.com}
{https://orcid.org/0000-0002-7497-5834}
{Israel Science Foundation grant no. 712/20, 1092/24.}

\Copyright{Samuel Bismuth and Erel Segal-Halevi}

\authorrunning{Bismuth and Segal-Halevi}


\category{}

\nolinenumbers 

\usepackage[vskip=2mm,font=itshape]{quoting}
\usepackage{algorithm,algpseudocode}
\usepackage{changepage}   

\newcommand{\len}{\operatorname{len}}

\newcommand{\bx}{\mathbf{x}}
\newcommand{\by}{\mathbf{y}}
\newcommand{\br}{\mathbf{r}}

\newcommand{\optint}{\ensuremath{\operatorname{OptInt}}}
\newcommand{\optfrac}{\ensuremath{\operatorname{OptFrac}}}
\newcommand{\rep}[2]{\ensuremath{{#1}^{* #2}}}
\newcommand{\fint}{\ensuremath{\operatorname{F}}}

\newcommand{\conv}{\ensuremath{\operatorname{conv}}}

\newcommand{\Q}{\mathbb{Q}}
\newcommand{\Z}{\mathbb{Z}}

\newtheorem{open}{Open Question}

\begin{document}

\fontsize{11}{13.5}\selectfont 

\maketitle

\begin{abstract}
\fontsize{12}{13.5}\selectfont 
We present an approximation notion for NP-hard optimization problems. The notion is based on an \emph{amortized relaxation}: the relaxed optimum of an input is the largest per-copy value attainable when many copies of the input are solved together. We prove that (assuming P$\neq$NP) the new notion is strictly stronger than FPTAS, but strictly weaker than having a polynomial-time algorithm.
Our results introduce a new computational complexity class, which is a strict superset of P and a strict subset of FPTAS.
\end{abstract}
\newpage

\section{Introduction}
When an optimization problem is found to be {\sf NP}-hard, we  assume that it cannot be solved exactly by a polynomial-time algorithm, and look for polynomial-time approximation algorithms. The most efficient kind  of an approximation algorithm currently known is the FPTAS (Fully Polynomial Time Approximation Scheme): for any $\epsilon>0$, it finds a solution that is at least $(1-\epsilon)$ times the optimal solution (in case of a maximization problem),  and runs in time polynomial in the input size and $1/\epsilon$.
Schuurman and Woeginger write in their Approximation Schemes Tutorial \cite{schuurman2001approximation} that
\begin{quote}
``With respect to worst case approximation, an FPTAS is the strongest possible result that we
can derive for an NP-hard problem''.	
\end{quote}
In this paper we challenge this claim. We show a new (to the best of our knowledge) kind of approximation algorithm, that is stronger than FPTAS in a precise sense. We call it FFPTAS --- Fractional Fully Polynomial Time Approximation Scheme.

Instead of approximating the optimal solution value, an FFPTAS approximates the \emph{amortized optimum} of the problem: the largest per-copy value that can be attained when many independent copies of the input are solved together as a single instance. The amortized optimum is a relaxation of the optimum --- it is always at least as large --- and it plays the role that the optimum of a fractional relaxation plays in classical approximation algorithms.
For any $t>0$, it finds a solution with value at least $(1-t)$ times the amortized optimum (or asserts that such a solution does not exist), and runs in time polynomial in the input size and $1/t$.

We prove that (under certain conditions) FFPTAS is a strictly better approximation than FPTAS, that is: every problem that has an FFPTAS has an FPTAS, but the opposite is not true unless {\sf P=NP}. 
We complement this result by showing an FFPTAS for an NP-hard problem. 
Together, our results assert the existence of a new complexity class for optimization problems, that lies strictly between {\sf P} and {\sf FPTAS}.

\subsection{Motivation}
Our main motivation for studying FFPTAS is the theoretical discovery of a new complexity class, refining the complexity hierarchy of NP-hard optimization problems.

But FFPTAS might also have practical applications in situations in which a perfect solution is required, but can be attained only when allowing fractions.
As an example, consider the problem of dividing items of different values between two partners. It may be required by law to give each partner exactly $1/2$ of the total value. However, such a perfect partition may be impossible to attain if the items cannot be split.
A possible solution is to have the partner who received the higher value compensate the other partner by monetary payments; but the amount of monetary payments available might also be bounded.
For this problem, the amortized optimum equals exactly the perfect partition value (half the total value; see the proof of Part (2) of \Cref{thm:main}).
If there exists a partition in which the smallest sum is at least $(1-t)\cdot \optfrac$, where $\optfrac$ is the perfect partition value, then it is possible to balance the partition by a payment of size at most $t\cdot \optfrac$.

An FFPTAS may also be useful as a subroutine for solving more complex problems.
As an example, consider the problem of number partitioning with bounded splitting \cite{bismuth2024partitioning}.
In this problem, we have to compute a perfect fractional partition of items among $n$ bins, while splitting at most $s$ items. With $s=n-1$ a perfect partition always exists and is easy to find, but with $s<n-1$ a perfect partition might not exist. To decide if such a partition exists, we need to know if the optimal integral solution (with no splittings) is sufficiently close to the optimal fractional solution, such that the difference can be completed with fractions of a small number of items. 
We need to do this without explicitly computing the optimal integral solution, which is usually NP-hard.
This is exactly what an FFPTAS does. 
It is also clear here why an ordinary FPTAS would not suffice: an FPTAS only tells us how close a solution is to the discrete optimum, whereas what we need to know is how far the discrete optimum is from the fractional optimum.

\section{Definitions}
We state all our definitions and results in terms of maximization problems. We have analogous results for minimization problems.

\begin{definition}[Maximization problem]
\label{def:problem}
A \emph{maximization problem} is defined by:

(a) An input domain $D$ --- a set of the form $D = \bigcup\limits_{m\geq 1}X^m$ for some base set $X$;

(b) A function $\optint:~  D \to \Q_{\geq 0}$ 	of the form
\begin{align}
\label{eq:combinatorial-formulation}
\optint(\bx) := \max_{\by} g(\bx,\by) \text{~subject to~} \by\in \fint(\bx),
\end{align}
where $g:D \times D\to \Q_{\geq 0}$ is the \emph{objective function} and $\fint :D \to \mathcal{P}(D)$ is the set of  solutions $\by$ considered feasible for input $\bx$. 	
\end{definition}
In this paper usually $X = \Z_{> 0}$ --- the set of positive integers, so that $D$ is the set of all finite vectors of positive integers.
~
For an input $\bx\in D$, the value $\optint(\bx)$ is called the \emph{optimal value} of $\bx$.	

Note that an optimization problem can have many different representations, that is, there can be many different $g$ and $\fint$ that yield the same $\optint$ function. This does not affect our results, as our results are defined directly using the function $\optint$.

\begin{definition}
\label{def:optfrac}
(a)
For a vector $\bx = (x_1,\ldots,x_m)\in X^{m}$ and an integer $k\geq 1$, we denote by $\rep{\bx}{k} \in X^{k m}$ the \emph{$k$-fold replication} of $\bx$: the concatenation of $k$ copies of $\bx$,
\begin{align*}
\rep{\bx}{k} := (\underbrace{x_1,\ldots,x_m,~ x_1,\ldots,x_m,~ \ldots,~ x_1,\ldots,x_m}_{k \text{ copies}}).
\end{align*}

(b)
For a maximization problem $\optint$ and an input $\bx\in D$, the \emph{amortized optimum} (or \emph{relaxed optimum}) of $\bx$ is:
\begin{align*}
\optfrac(\bx) := \sup_{k \geq 1} \frac{\optint(\rep{\bx}{k})}{k}.
\end{align*}
\end{definition}

Taking $k=1$ in the supremum shows that, for any maximization problem,
\begin{align}
\label{eq:frac-geq-int}
\optfrac(\bx)\geq \optint(\bx)\geq 0 &&\text{for all $\bx\in D$}.
\end{align}
In this paper, we consider only problems for which $\optfrac(\bx)$ (hence also $\optint(\bx))$) is finite for all $\bx\in D$.

The intuition behind \Cref{def:optfrac} is that solving many copies of an instance together may attain a higher value per copy than solving each copy separately, since the discrete solutions of the different copies can complement each other. The per-copy value in the limit plays the role of a fractional optimum.
This is analogous to the way fractional parameters are defined in graph theory: the fractional chromatic number \cite{hilton1973colour,scheinerman1997fractional} is the limit of $\chi(G[k])/k$ over blow-ups $G[k]$ of the graph $G$, and the Shannon capacity of a graph \cite{shannon1956zero} is the limit of $\alpha(G^{\boxtimes k})^{1/k}$, where $G^{\boxtimes k}$ denotes the $k$-times strong power of $G$, and $\alpha(\cdot)$ denotes its independence number. In both cases, a ``fractional'' quantity is defined by amortizing an integral quantity over replicated instances.%
\footnote{
We remark that if $\optint$ is \emph{superadditive under replication}, that is, $\optint(\rep{\bx}{(j+k)}) \geq \optint(\rep{\bx}{j}) + \optint(\rep{\bx}{k})$ for all $j,k\geq 1$, then by Fekete's lemma \cite{fekete1923verteilung} the supremum in \Cref{def:optfrac} equals the limit $\lim_{k\to\infty} \optint(\rep{\bx}{k})/k$. This holds for the partition problems studied in this paper. Our results do not require superadditivity, so we work with the supremum.
}

We now define several classes of algorithms for maximization problems.
We denote by $\len(\bx)$ the number of bits in the binary representation of the vector $\bx$.

\begin{definition}
\label{def:algorithms}
Let $\optint$ be a maximization problem.
\begin{itemize}
\item A \emph{polynomial-time algorithm} for $\optint$ is an algorithm that accepts as input a vector $\bx\in D$ and returns as output the value $\optint(\bx)$, and runs in time polynomial in $\len(\bx)$.

\item A \emph{polynomial-time relaxation algorithm} for $\optint$ is an algorithm that accepts as input a vector $\bx\in D$ and returns as output the value $\optfrac(\bx)$, and runs in time polynomial in $\len(\bx)$.%
\footnote{
	There are many problems for which computing $\optint$ is NP-hard whereas computing $\optfrac$ is in \sf{P}; the partition problems in the proofs of Parts (2) and (4) of \Cref{thm:main} are examples.
	In  \Cref{sec:rel-harder-than-opt} we show that the converse is also possible: there are optimization problems in which computing $\optfrac$ is NP-hard whereas computing $\optint$ is in \sf{P}. 
}

\item An \emph{FPTAS} for $\optint$, denoted \emph{FPTAS[$\optint$]}, is an algorithm that, 
given input vector $\bx\in D$ and real number $\epsilon>0$,
runs in time polynomial in $\len(\bx)$ and $1/\epsilon$, and 
returns a value $v$ such that $(1 - \epsilon)\cdot \optint(\bx) \leq v \leq \optint(\bx)$.

\item An \emph{FFPTAS} for $\optint$, denoted \emph{FFPTAS[$\optint$]}, is an algorithm that, 
given input vector $\bx\in D$ and real number $t>0$,
runs in time polynomial in $\len(\bx)$ and $1/t$, and 
returns a value $v$ such that $(1 - t)\cdot \optfrac(\bx) \leq v \leq \optint(\bx)$ (if such a value exists, that is, if $\optint(\bx) \geq (1-t)\cdot\optfrac(\bx)$), or None (if no such value exists, that is, if $\optint(\bx) < (1-t)\cdot\optfrac(\bx)$).
\end{itemize}
\end{definition}






\begin{remark}
\label{rem:feasible-value}
In \Cref{def:algorithms}, the requirement $v\leq \optint(\bx)$ plays the role that ``$v$ is the value of a feasible solution'' plays in the classical definition of an FPTAS: since $\optint(\bx)$ is the optimal value, the value of any feasible solution is at most $\optint(\bx)$. When the problem arises from a combinatorial formulation such as \eqref{eq:combinatorial-formulation}, the definitions of FPTAS and FFPTAS can be strengthened such that the algorithms must return the value of an actual integral solution (and even the solution itself); the algorithms we construct in the proofs of \Cref{thm:main} satisfy this stronger requirement.
\end{remark}

\begin{definition}
A problem $\optint$ is called \emph{fractionally-polynomial} if it satisfies the following two conditions:

1. $\optfrac(\bx)$ can be computed in time polynomial in $\len(\bx)$ (in particular, it is finite and rational).

2. $\optint(\bx)$ can be \emph{represented} in space polynomial in $\len(\bx)$, that is, $\len(\optint(\bx))\leq  p(\len(\bx))$ for some polynomial function $p$, and $p$ itself can be computed in time polynomial in $\len(\bx)$.
\end{definition}

\section{Main theorem}

\begin{theorem}
\label{thm:main}
Let FracP be the set of fractionally-polynomial maximization problems.

(1) If a problem in FracP has a poly-time algorithm, then it has an FFPTAS.

(2) Some problem in FracP has an FFPTAS but no poly-time algorithm if {\sf P} $\neq$ {\sf NP}. 

(3) If a problem in FracP has an FFPTAS, then it has an FPTAS. 

(4) Some problem in FracP has an FPTAS but has no FFPTAS
unless {\sf P} $=$ {\sf NP}.
\end{theorem}
\Cref{thm:main} shows that, at least among the class of fractionally-polynomial problems, FFPTAS is a strictly stronger approximation notion than FPTAS, but strictly weaker than {\sf P}.

\subsection{Proof of part (1)}
\begin{claim*}[Part (1)]
Let $\optint$ be a fractionally-polynomial problem.
If there exists a polynomial-time algorithm for computing  $\optint(\bx)$, then $\optint$ has an FFPTAS.
\end{claim*}
\begin{proof}
The claim is easily proved by \Cref{alg:generalffptas-max} below. 
\begin{algorithm}[!h] \caption{\qquad FFPTAS$[\optint](\bx,t)$} \label{alg:generalffptas-max}
\begin{algorithmic}[1]
	\State Compute $\optfrac(\bx)$ and $\optint(\bx)$.
	\If {$\optint(\bx) \geq (1-t)\optfrac(\bx)$} \Return $\optint(\bx)$;
	\Else {~($\optint(\bx) < (1-t)\optfrac(\bx)$)} \Return None.
	\EndIf
\end{algorithmic} 
\end{algorithm}	

Correctness is immediate by definition: if $\optint(\bx) \geq (1-t)\optfrac(\bx)$, then the value $v := \optint(\bx)$ satisfies $(1-t)\optfrac(\bx) \leq v \leq \optint(\bx)$, as required by the definition of FFPTAS; otherwise, None is the required output.

The computations in Step 1
take time polynomial in $\len(\bx)$. This holds for $\optfrac(\bx)$ since $\optint$ is fractionally-polynomial, and holds for $\optint(\bx)$  by the claim assumption.
The comparisons in Steps 2 and 3 require an additional time that is linear in $\len(t)$, which is in $O(\log(1/t))$, so the run-time fulfills the FFPTAS requirements.
\end{proof}

\subsection{Proof of part (2)}
\label{sec:part-2}
\begin{claim*}[Part (2)]
There exists a fractionally-polynomial problem $\optint$ such that $\optint$ has an FFPTAS, but has no polynomial-time algorithm unless {\sf P} $=$ {\sf NP}. 
\end{claim*}

We use the max-min variant of the {\sc Partition} problem with two bins.
The input is a list of some $m$ items with positive integer sizes; the goal is to partition the list  into two bins such that the sum of item sizes in the bin with the smallest sum is as large as possible. In our notation, it can be presented as follows.
\begin{itemize}
\item The input domain $D$ is the set of all vectors $\bx = (x_1,\ldots,x_m)$ of positive integers, for all $m\geq 1$, where $x_i$  represents the size of item $i$. 
\item The optimal value is
\begin{align*}
	\optint(\bx) := \max_{T\subseteq \{1,\ldots,m\}} ~\min\Big(\sum_{i\in T} x_i,~ \sum_{i\notin T} x_i\Big),
\end{align*}
where $T$ is the set of items in bin \#1: the objective is to maximize the minimum bin sum.
\end{itemize}
We denote this problem by ${\cal P}_2$, and denote $Z := \sum_{i=1}^m x_i$ (the sum of all item sizes).

We first compute the amortized optimum of ${\cal P}_2$.
\begin{lemma}
\label{claim:p2-optfrac}
$\optfrac(\bx) = Z / 2$
for every input $\bx\in D$.
\end{lemma}
\begin{proof}
\emph{Upper bound}: for every $k\geq 1$, the total item size in $\rep{\bx}{k}$ is $kZ$, and in every partition into two bins, the minimum bin sum is at most the average bin sum, which is $kZ/2$. Hence $\optint(\rep{\bx}{k}) \leq kZ/2$, so $\optint(\rep{\bx}{k})/k \leq Z/2$ for every $k$.

\emph{Lower bound}: for $k=2$, the instance $\rep{\bx}{2}$ consists of two disjoint copies of the item list. Putting the first copy in bin \#1 and the second copy in bin \#2 yields a partition in which both bin sums equal $Z$, so $\optint(\rep{\bx}{2}) \geq Z$, that is, $\optint(\rep{\bx}{2})/2 \geq Z/2$.

Together, $\optfrac(\bx) = \sup_k \optint(\rep{\bx}{k})/k = Z/2$, attained at $k=2$.
\end{proof}

\begin{proof}[Proof of Part (2)]
The problem ${\cal P}_2$ is fractionally-polynomial: 
$\optint(\bx)$ is a sum of some subset of the integers in $\bx$, so it can be represented in space polynomial in $\len(\bx)$; and by \Cref{claim:p2-optfrac}, $\optfrac(\bx) = Z/2$, which can be computed in time polynomial in $\len(\bx)$. Note also that $\optfrac(\bx) = Z/2 > 0$, since the item sizes are positive integers.

It is well known that ${\cal P}_2$ is NP-hard but has an FPTAS. We now show that it also has an FFPTAS.

We will use an auxiliary combinatorial problem --- a variant of ${\cal P}_2$ with a \emph{critical coordinate}, which we denote by ${\cal P}_{2,cc}$.
(The auxiliary problem is used only as a subroutine, so we describe it directly in combinatorial terms.)
The input to ${\cal P}_{2,cc}$
is a vector $(x_1,\ldots,x_m, z)$, where the $x_i$ are the item sizes and $z$ is a rational number representing a lower bound on the sum of bin \#1. 
A feasible solution is a partition of the items into two bins such that the sum of bin \#1 is at least $z$; the objective is to maximize the sum of bin \#2.
Note that whenever $z\leq Z$, a feasible solution exists (put all items in bin \#1).

The general scheme of Woeginger \cite{woeginger2000does} can be used to construct an FPTAS for ${\cal P}_{2,cc}$
(in fact, ${\cal P}_{2,cc}$ is equivalent to the {\sc Subset Sum} problem:
the objective is to choose a subset of items for bin \#2 with a maximum sum, subject to that sum being at most $Z - z$).
We assume that FPTAS[${\cal P}_{2,cc}$] returns the partition it computes, and not only its value (Woeginger's scheme does).
\Cref{alg:fptad-max} below uses this FPTAS to design an FFPTAS for 
${\cal P}_{2}$.
\begin{algorithm}[!h] \caption{\qquad FFPTAS$[{\cal P}_2](\bx,t)$} \label{alg:fptad-max}
\begin{algorithmic}[1]
	\State 
	Compute $\optfrac(\bx) = Z/2$ and let $z := {(1-t)\cdot \optfrac(\bx)}$.
	\State 
	Run FPTAS$[{\cal P}_{2,cc}]$
	with input $(x_1,\ldots,x_m,z)$ and approximation accuracy $\epsilon = t/2$.
	\State 
	Denote by $v$ the value returned by the FPTAS (the sum of bin \#2 in the computed partition; the sum of bin \#1 in that partition is $Z - v \geq z$).
	\If {\label{line:return-yes} $v \ge z$}
	\Return $\min(v, Z - v)$.
	\Else{~($v < z$)} \Return None.
	\EndIf
\end{algorithmic} 
\end{algorithm}	


The run-time of \Cref{alg:fptad-max} is dominated by the run-time of FPTAS[${\cal P}_{2,cc}$], which is polynomial in $\len(\bx)$ and $1/\epsilon = 2/t$.
It remains to prove that \Cref{alg:fptad-max} is indeed an FFPTAS for ${\cal P}_{2}$.
\begin{itemize}
\item 
Suppose the algorithm returns a value $w = \min(v,Z-v)$ in Line \ref{line:return-yes}. The partition computed by FPTAS[${\cal P}_{2,cc}$] is a feasible solution of ${\cal P}_{2,cc}$, so the sum of its bin \#1 is $Z - v \geq z$; and the condition of Line \ref{line:return-yes} gives $v\geq z$. Hence $w = \min(v,Z-v)\geq z= (1-t)\optfrac(\bx)$. Moreover, $w$ is the minimum bin sum of an actual partition of $(x_1,\ldots,x_m)$, so $w\leq \optint(\bx)$. Hence, $w$ is a correct value for the FFPTAS.
\item 
Conversely, suppose that $\optint(\bx)\geq z$, that is, there exists a partition of $(x_1,\ldots,x_m)$ in which the sum of both bins is at least $z$. By symmetry, we can assume that bin \#1 has the smaller sum, so the sum of bin \#2 is at least half the sum of all items, that is, at least $Z/2 = \optfrac(\bx)$.

That partition is a feasible solution to ${\cal P}_{2,cc}$, as the sum of bin \#1 is at least $z$. Its objective value for that problem is the sum of bin \#2, which is at least $\optfrac(\bx)$.
Therefore, FPTAS[${\cal P}_{2,cc}$] must output a solution with objective value $v\geq (1-\epsilon)\cdot \optfrac(\bx) > (1-t)\cdot\optfrac(\bx) = z$ (the strict inequality holds since $\epsilon = t/2 < t$ and $\optfrac(\bx) > 0$).
Hence, the algorithm returns a value in Line \ref{line:return-yes}, as required.
\end{itemize}
So  \Cref{alg:fptad-max} returns a value $w$ with $(1-t)\optfrac(\bx)\leq w \leq \optint(\bx)$ if $\optint(\bx)\geq (1-t)\optfrac(\bx)$, and returns None otherwise (by the first bullet, whenever a value is returned, $\optint(\bx)\geq w \geq (1-t)\optfrac(\bx)$; so if $\optint(\bx) < (1-t)\optfrac(\bx)$, no value can be returned). Hence, it is an FFPTAS for  ${\cal P}_2$.

On the other hand, it is known that the max-min {\sc Partition} problem (even with $n=2$ bins) is {\sf NP}-hard, so, there is no polynomial time algorithm for ${\cal P}_2$ unless {\sf P} $=$ {\sf NP}. 
This concludes the proof of Part (2).
\end{proof}

\subsection{Proof of part (3)}
\begin{claim*}[Part (3)]
Let $\optint$ be a fractionally-polynomial problem.
If $\optint$ has an FFPTAS, then it has an FPTAS.
\end{claim*}

\begin{proof}
We prove the claim by constructing an FPTAS for $\optint$.

We first use the given FFPTAS$[\optint]$ with parameter $t=\epsilon$, to look for a value $v$ with $(1-\epsilon)\cdot \optfrac(\bx)\leq v\leq \optint(\bx)$. If the FFPTAS finds such a value, then we return it as the result of the FPTAS: by \eqref{eq:frac-geq-int}, $v\geq (1-\epsilon)\optfrac(\bx) \geq (1-\epsilon)\optint(\bx)$, and $v\leq \optint(\bx)$, as required.

If the FFPTAS returns None, we know that $\optint(\bx) < (1-\epsilon)\cdot \optfrac(\bx)$. 
We use FFPTAS$[\optint]$ again 
with $t=1-(1-\epsilon)^2$, to look for a value $v$ with $(1-\epsilon)^2\cdot \optfrac(\bx)\leq v\leq \optint(\bx)$. 
If the FFPTAS finds such a value, then we return it, as it is at least $(1-\epsilon)^2 \optfrac(\bx)$, which is larger than $(1-\epsilon)\optint(\bx)$ by the inequality starting this paragraph.

We proceed similarly, using FFPTAS$[\optint]$ to look for a value  $v$ with $(1-\epsilon)^k \cdot \optfrac(\bx)\leq v\leq \optint(\bx)$, for increasing values of $k$. If such a value is found, then we return it: for $k\geq 2$, the FFPTAS returned None at iteration $k-1$, so $\optint(\bx) < (1-\epsilon)^{k-1}\optfrac(\bx)$, and hence $v\geq (1-\epsilon)^k\optfrac(\bx) > (1-\epsilon)\optint(\bx)$.
If a value is not found, we know that
$\optint(\bx)/\optfrac(\bx) < (1-\epsilon)^k$.

We now show that, after a polynomial number of iterations, if a solution has not been found yet, we can stop and return $0$.

Because $\optint$ is fractionally-polynomial, both $\len(\optint(\bx))$ and $\len(\optfrac(\bx))$ are polynomial in $\len(\bx)$. Hence, 
$\len(\optint(\bx) / \optfrac(\bx))$ is 
at most $q(\len(\bx))$, where $q$ is some polynomial function. 
In particular, the denominator of the ratio $\optint(\bx) / \optfrac(\bx)$ is at most $2^{q(\len(\bx))}$.
Hence, if $\optint(\bx) / \optfrac(\bx)  < 1/2^{q(\len(\bx))}$, then necessarily $\optint(\bx)=0$, and the FPTAS can return $0$ (which satisfies $(1-\epsilon)\optint(\bx)\leq 0\leq \optint(\bx)$). 

We now solve the inequality to find $k$:
\begin{align*}
(1-\epsilon)^k &< 1/2^{q(\len(\bx))}
\\
[1/(1-\epsilon)]^k &> 2^{q(\len(\bx))}
\\
k &> q(\len(\bx)) / \log_2[1/(1-\epsilon)]
\end{align*}
As $\log[1/(1-\epsilon)] = \sum_{i=1}^\infty \frac{\epsilon^i}{i} > \epsilon$, it is sufficient to run the algorithm for all $k$ between $1$ and $q(\len(\bx))/\epsilon$; if no value is found, the FPTAS can return $0$. 
The pseudo-code is shown in \Cref{alg:generalfptas-max} below. 

\begin{algorithm}[!h] \caption{\qquad FPTAS$[\optint](\bx,\epsilon)$} \label{alg:generalfptas-max}
\begin{algorithmic}[1]
	\State\label{step:compute-max} Compute $\optfrac(\bx)$. {\sf If} $\optfrac(\bx)=0$ {\sf then} \Return $0$. \Comment{by \eqref{eq:frac-geq-int}, $\optint(\bx)=0$ in this case}
	\State Let $q$ be the polynomial function such that the denominator of $\optint(\bx)/\optfrac(\bx)$ is at most $2^{q(\len(\bx))}$ for all $\bx\in D$.
	\For{$k := 1$ to $q(\len(\bx))/\epsilon$}
	\State \label{step:t} 
	$t \longleftarrow 1 - (1 - \epsilon)^k$.
	\State \label{step:yes-max} {\sf If} FFPTAS$[\optint](\bx,t)$ returns a value $v$, {\sf then} \Return $v$.
	\EndFor
	\State 
	\Return $0$.
\end{algorithmic} 
\end{algorithm}	

The number of calls to FFPTAS$[\optint]$ is polynomial in $\len(\bx)$ and $1/\epsilon$. In each call, the parameter $t$ is at least $\epsilon$, and the runtime of FFPTAS$[\optint]$ is polynomial in $\len(\bx)$ and $1/t\leq 1/\epsilon$; hence the total run-time of the algorithm is polynomial in $\len(\bx)$ and $1/\epsilon$, as required by the definition of FPTAS.
\end{proof}

\subsection{Proof of part (4)}
\begin{claim*}[Part (4)]
There exists a fractionally-polynomial problem $\optint$ such that $\optint$ has an FPTAS, but has no FFPTAS unless {\sf P} $=$ {\sf NP}. 
\end{claim*}
\begin{proof}
We use the max-min variant of 4-way {\sc Partition}, which we denote by ${\cal P}_4$.%
\footnote{
We have a proof of the same claim using 
3-way {\sc Partition}. We use
4-way {\sc Partition} here because the proof is slightly simpler.
}
${\cal P}_4$ is defined similarly to ${\cal P}_2$ from Part (2), except that the items are partitioned into four bins: the set of valid inputs $D$ is the same, and the optimal value is
\begin{align*}
\optint(\bx) := \max_{b:\{1,\ldots,m\}\to\{1,2,3,4\}} ~\min_{j\in\{1,2,3,4\}} \sum_{i:\, b(i)=j} x_i,
\end{align*}
where $b$ assigns each item to a bin: the objective is to maximize the minimum bin sum. As before, we denote $Z := \sum_{i=1}^m x_i$.

We first compute the amortized optimum of ${\cal P}_4$, analogously to \Cref{claim:p2-optfrac}:
\begin{itemize}
\item[--] \emph{Upper bound}: for every $k\geq 1$, the total item size in $\rep{\bx}{k}$ is $kZ$, and in every partition into four bins, the minimum bin sum is at most the average bin sum, which is $kZ/4$. Hence $\optint(\rep{\bx}{k})/k \leq Z/4$ for every $k$.
\item[--] \emph{Lower bound}: for $k=4$, the instance $\rep{\bx}{4}$ consists of four disjoint copies of the item list. Putting copy $j$ in bin $j$ (for $j\in\{1,2,3,4\}$) yields a partition in which every bin sum equals $Z$, so $\optint(\rep{\bx}{4})/4 \geq Z/4$.
\end{itemize}
Hence $\optfrac(\bx) = Z/4$, and the supremum is attained at $k=4$.
As $\optfrac(\bx) = Z/4$ can be computed in polynomial time, and $\optint(\bx)$ is a sum of a subset of the integers in $\bx$, the problem is fractionally-polynomial.
~
Woeginger \cite{woeginger2000does} gave an FPTAS for ${\cal P}_4$.%
\footnote{More generally, he gave an FPTAS for $k$-way partition for any fixed $k$. Woeginger's FPTAS returns the value of an actual partition, which is at most $\optint(\bx)$, so it satisfies our definition of FPTAS.}
We now prove that 4-way partition has no FFPTAS unless {\sf P} $=$ {\sf NP}. 

The proof is by reduction from the 
{\sc Equal-Cardinality Partition} problem:
given a list with $2 m$ integers and sum $2 S$, decide if they can be partitioned into two subsets with cardinality $m$ and sum $S$.
It is  proved to be {\sf NP}-hard in \cite{gareycomputers}. 

Given an instance ${\cal X}_1$ of {\sc Equal-Cardinality Partition}, we can assume w.l.o.g. that all integers in ${\cal X}_1$ are at most $S$, since otherwise the answer is necessarily ``no''.

We construct an instance ${\cal X}_2$ of the {\sc Equal-Cardinality Partition} problem by replacing each integer $\bx$ in ${\cal X}_1$ by $2 \bx + 4 S$. 
So ${\cal X}_2$ contains $2 m$ integers between $4 S$ and $6 S$.
Their sum is $2\cdot (2 S) + 2 m\cdot (4 S) = 2(4m+2)S$.
Clearly, ${\cal X}_1$ has an equal-sum equal-cardinality partition (with bin sums $S$) if and only if ${\cal X}_2$ has an equal-sum equal-cardinality partition (with bin sums $(4 m + 2) S$).

We construct an instance ${\cal X}_4$ for ${\cal P}_{4}$ that contains $2 (2m+1)$ items: The $2 m$ items in ${\cal X}_2$, and $2(m+1)$ additional \emph{small items} with size $4 S$.
So the sum of all item sizes in ${\cal X}_4$ is:
\begin{align*}
2 (4m + 2) \cdot S  + 2 (m+1) \cdot 4 S
=
4 (4m + 3) \cdot S.
\end{align*}
Hence, $\optfrac({\cal X}_4) = (4m+3)S$ (a quarter of the total item size, by the computation of $\optfrac$ above).

\begin{claim}
\label{claim:part-4-aux}
${\cal X}_4$ can be partitioned into four bins with sum at least $(4m+2)S$,
if-and-only-if
${\cal X}_2$ can be partitioned into two bins with cardinality $m$ and sum exactly $(4m+2)S$.
\end{claim}
\begin{proof}
$\impliedby$:
Suppose ${\cal X}_2$ can be partitioned into two bins of sum $(4m+2)S$. 
The $2 (m+1)$ small items can be divided into two additional bins of $m+1$ items each, with sum 
\begin{align}
	\label{eq:(m-1)nS-max}
	(m+1) \cdot 4 S = (4m+4)S > (4m+2)S,
\end{align}
so in the resulting partition of ${\cal X}_4$, the sum of each of the four bins is at least $(4m+2)S$.

$\implies$:
Suppose ${\cal X}_4$ can be partitioned into four bins with sum at least $(4m+2)S$. Let us analyze the structure of this partition.
\begin{itemize}
\item [--] Since the sum of each bin is at least 
$(4m+2)S$, and the average sum of a bin is $(4m+3)S$, the sum of every two bins is at most 
$4(4m+3)S - 2(4m+2)S = 2(4m+4)S$.

\item [--] The sum of the smallest $2(m+1)$ items in ${\cal X}_4$ is exactly $2(4m+4)S$. Hence, every two bins must contain together at most $2(m+1)$ items.
	
\item [--] Let $c_1,c_2,c_3,c_4$ be the number of items in the four bins, so that $c_1+c_2+c_3+c_4 = 2m + 2(m+1) = 4m+2$. By the previous item, $c_i + c_j \le 2(m+1)$ for every pair of bins $i\neq j$. We claim that some pair attains this bound with equality.
Suppose not, so that $c_i + c_j \le 2m+1$ for every pair. 
Indeed, there are three ways to split the four bins into two disjoint pairs; in each such split the two pair-sums add up to $4m+2$, while each pair-sum is at most $2m+1$, so each pair-sum equals exactly $2m+1$. Applying this to the splits $\{1,2\}\mid\{3,4\}$, $\{1,3\}\mid\{2,4\}$, and $\{1,4\}\mid\{2,3\}$ gives $c_1+c_2 = c_1+c_3 = c_1+c_4 = 2m+1$, whence $c_2 = c_3 = c_4 =: c$. Then $c_1 + 3c = 4m+2$ and $c_1 + c = 2m+1$, so $2c = 2m+1$ which is impossible for an integer $c$. Hence some two bins contain together {exactly} $2(m+1)$ items; w.l.o.g.\ assume that these are bins \#1 and \#2.
	
\item [--]
The sum of bins \#1 and \#2 is at most $2(4m+4)S$, they contain exactly $2(m+1)$ items, and each item size is at least $4 S$. Therefore, bins \#1 and \#2 must contain only items of size exactly $4 S$, and we can assume w.l.o.g. that these are the $2(m+1)$ \emph{small} items.
	\item [--] The other two bins, bin \#3 and bin \#4, contain together the $2m$ items of ${\cal X}_2$, their sum is $2(4m+2)S$, and the sum of each bin is at least $(4m+2)S$, so the sum of each bin must in fact be exactly $(4m+2)S$. 
	But the sum of every $m+1$ items is at least $(m+1)4S > (4m+2)S$. Hence, bins \#3 and \#4 must contain exactly $m$ items each.
\end{itemize}   
This concludes the proof of the auxiliary claim.
\end{proof}

We now return to the proof of the theorem.
We apply FFPTAS$[{\cal P}_4]$ on ${\cal X}_4$ with parameter $t := 1/(4m+3)$.
Note that
$$
(1-t)\cdot \optfrac({\cal X}_4) = 
(1-t)\cdot (4m+3)S = 
(4m+2)S.
$$
By definition, FFPTAS$[{\cal P}_4]$ returns a non-None value if and only if $\optint({\cal X}_4)\geq (4m+2)S$, that is, if and only if there exists a partition of ${\cal X}_4$ into four bins with sum at least $(4m+2)S$.
By the auxiliary claim, 
this holds if and only if 
the {\sc Equal-Cardinality Partition} instance $({\cal X}_2)$ is a ``yes'' instance. 
As $1/t = 4m+3$ is polynomial in the input size,
FFPTAS$[{\cal P}_4]$ 
could be used to solve 
{\sc Equal-Cardinality Partition} in polynomial time, which is impossible unless {\sf P} $=$ {\sf NP}.
\end{proof}

\section{Implications and Discussion}
\label{sec:future}
\begin{definition}
{\sf FFPTAS} is the class of problems that have an FFPTAS
\end{definition}

The most interesting implication of \Cref{thm:main} is that {\sf FFPTAS} is a new complexity class, as
\begin{align*}
({\sf P}\cap {\sf FracP})  ~~\subsetneq~~  ({\sf FFPTAS}\cap {\sf FracP})  ~~\subsetneq~~  ({\sf FPTAS}\cap {\sf FracP}).
\end{align*}

This implication hinges on the fact that an optimization problem is defined by the function $\optint$, and the 
relaxation function $\optfrac$ is defined directly from $\optint$ (\Cref{def:optfrac}).
An alternative we considered is to define an optimization problem as a pair $(\fint,g)$, and define 
\begin{align*}
\optint(\bx) &:= \max_{\by} g(\bx,\by) \text{~subject to~} \by\in \fint(\bx),
\\	
\optfrac(\bx) &:= \max_{\by} g(\bx,\by) \text{~subject to~} \by\in \conv(\fint(\bx)),
\end{align*}
where $\conv$ denotes the convex hull of a set of vectors.
\Cref{thm:main} can be proved for this definition too; however, it does not define a new complexity class, as the same optimization problem $\optint$ can have different representations as a pair $(\fint,g)$, which lead to different $\optfrac$.%
\footnote{
A famous example is bin-packing: its naive representation has a decision variable for each (item,bin) pair, but it has a more sophisticated representation known as the \emph{configuration program}, whose decision variables correspond to bin configurations. The convex hulls of the integer solution sets of these two representations are substantially different, and the relaxation of the configuration representation is much more useful for approximation algorithms, as shown by \cite{fernandez1981bin} and \cite{karmarkar1982efficient} and the many follow-up works on that problem.
} 
It is possible that an FFPTAS exists for some representations but not for others; hence, the complexity class ${\sf FFPTAS}$ would not be well-defined with this definition.

The restriction to fractionally-polynomial problems in \Cref{thm:main} is necessary: as shown in \Cref{sec:rel-harder-than-opt}, there is a problem that has a polynomial-time algorithm but no FFPTAS unless ${\sf P}={\sf NP}$, so ${\sf P}\not\subseteq {\sf FFPTAS}$ in general.
We can still show the existence of a class between {\sf P} and {\sf FPTAS} by defining
\begin{align*}
{\sf C} := {\sf P} ~\cup~ ({\sf FFPTAS}\cap \text{FracP}).
\end{align*}
Then,
\begin{align}
\label{eq:containment-c}
{\sf P} ~~\subsetneq~~  {\sf C}  ~~\subsetneq~~  {\sf FPTAS}:
\end{align}
we have ${\sf P} ~~\subseteq~~  {\sf C}$ by construction;
${\sf C}  ~~\subseteq~~  {\sf FPTAS}$
since ${\sf P}\subseteq {\sf FPTAS}$ and by \Cref{thm:main} part (3); 
the problem ${\cal P}_2$ of Part (2) witnesses the first strict containment; and the problem ${\cal P}_4$ of Part (4), which is in {\sf FPTAS} but neither in {\sf P} nor in {\sf FFPTAS}, witnesses the second.
One may ask whether there exists a class ${\sf C}$ satisfying \eqref{eq:containment-c} with 
a smoother definition.

Finally, to illustrate the generality of our definition and to inspire future work, we present in \Cref{sec:rsp} an NP-hard problem on graphs, which is fractionally-polynomial and has an FPTAS. We show that a special case of it, which is still NP-hard, has an FFPTAS, thus showing a second member in ${\sf FracP} \cap ({\sf FFPTAS} \setminus {\sf P})$. We leave as an open question whether it has an FFPTAS in general.

\section*{Acknowledgments}
Some of the technical results in the paper were developed in collaboration with claude.ai Opus 4.8 and Fable 5 models. 
Full details are available in the shared conversation link
\ifanon
, which is suppressed for anonymity.
\else
.\footnote{\url{https://claude.ai/share/22ea06cd-1dcd-4665-867e-f4ffc58fde33}}
We are grateful to 
Neal Young, Andras Farago,%
\footnote{See 
\url{https://cstheory.stackexchange.com/q/52830}
for the full discussion.
}
Edward A. Hirsch and Yossi Azar for insightful comments and discussions. 
\fi

\newpage
\bibliography{references}

\begin{thebibliography}{10}

\bibitem{bismuth2024partitioning}
Samuel Bismuth, Vladislav Makarov, Erel Segal-Halevi, and Dana Shapira.
\newblock Partitioning problems with splittings and interval targets.
\newblock In {\em 35th International Symposium on Algorithms and Computation
  (ISAAC 2024)}, pages 12--1. Schloss Dagstuhl--Leibniz-Zentrum f{\"u}r
  Informatik, 2024.

\bibitem{fekete1923verteilung}
Michael Fekete.
\newblock {\"U}ber die {V}erteilung der {W}urzeln bei gewissen algebraischen
  {G}leichungen mit ganzzahligen {K}oeffizienten.
\newblock {\em Mathematische Zeitschrift}, 17:228--249, 1923.

\bibitem{fernandez1981bin}
W~Fernandez~de La~Vega and George~S. Lueker.
\newblock Bin packing can be solved within 1+ $\varepsilon$ in linear time.
\newblock {\em Combinatorica}, 1(4):349--355, 1981.

\bibitem{gareycomputers}
M.~R. Garey and David~S. Johnson.
\newblock {\em Computers and Intractability: {A} Guide to the Theory of
  NP-Completeness}.
\newblock W. H. Freeman, 1979.

\bibitem{handler1980dual}
Gabriel~Y. Handler and Israel Zang.
\newblock A dual algorithm for the constrained shortest path problem.
\newblock {\em Networks}, 10(4):293--309, 1980.

\bibitem{hassin1992approximation}
Refael Hassin.
\newblock Approximation schemes for the restricted shortest path problem.
\newblock {\em Mathematics of Operations Research}, 17(1):36--42, 1992.

\bibitem{hilton1973colour}
A.~J.~W. Hilton, R.~Rado, and S.~H. Scott.
\newblock A ({$<$}5)-colour theorem for planar graphs.
\newblock {\em Bulletin of the London Mathematical Society}, 5(3):302--306,
  1973.

\bibitem{karmarkar1982efficient}
Narendra Karmarkar and Richard~M Karp.
\newblock An efficient approximation scheme for the one-dimensional bin-packing
  problem.
\newblock In {\em 23rd Annual Symposium on Foundations of Computer Science
  (sfcs 1982)}, pages 312--320. IEEE, 1982.

\bibitem{lorenz2001simple}
Dean~H. Lorenz and Danny Raz.
\newblock A simple efficient approximation scheme for the restricted shortest
  path problem.
\newblock {\em Operations Research Letters}, 28(5):213--219, 2001.

\bibitem{matousek2007understanding}
Jiri Matousek and Bernd G{\"a}rtner.
\newblock {\em Understanding and using linear programming}.
\newblock Springer Science \& Business Media, 2007.

\bibitem{mehlhorn2000resource}
Kurt Mehlhorn and Mark Ziegelmann.
\newblock Resource constrained shortest paths.
\newblock In {\em Proceedings of the 8th Annual European Symposium on
  Algorithms (ESA)}, volume 1879 of {\em LNCS}, pages 326--337. Springer, 2000.

\bibitem{motzkin1965maxima}
Theodore~S Motzkin and Ernst~G Straus.
\newblock Maxima for graphs and a new proof of a theorem of {Tur{\'a}n}.
\newblock {\em Canadian Journal of Mathematics}, 17:533--540, 1965.

\bibitem{scheinerman1997fractional}
Edward~R. Scheinerman and Daniel~H. Ullman.
\newblock {\em Fractional Graph Theory: A Rational Approach to the Theory of
  Graphs}.
\newblock Wiley, 1997.
\newblock Reprinted by Dover Publications, 2011.

\bibitem{schrijver1998theory}
Alexander Schrijver.
\newblock {\em Theory of linear and integer programming}.
\newblock John Wiley \& Sons, 1998.

\bibitem{schuurman2001approximation}
Petra Schuurman and Gerhard~J Woeginger.
\newblock Approximation schemes --- a tutorial, 2001.

\bibitem{shannon1956zero}
Claude~E. Shannon.
\newblock The zero error capacity of a noisy channel.
\newblock {\em IRE Transactions on Information Theory}, 2(3):8--19, 1956.

\bibitem{woeginger2000does}
Gerhard~J. Woeginger.
\newblock When does a dynamic programming formulation guarantee the existence
  of a fully polynomial time approximation scheme ({FPTAS})?
\newblock {\em {INFORMS} J. Comput.}, 12(1):57--74, 2000.

\end{thebibliography}

\newpage
\appendix
\section{Not Every Polynomial-time Solvable Problem is Fractionally-polynomial}
\label[appendix]{sec:rel-harder-than-opt}

For many combinatorial optimization problems, computing the optimal value $\optint(\bx)$ is NP-hard whereas computing the amortized optimum $\optfrac(\bx)$ can be done in polynomial time; the partition problems of \Cref{thm:main} are examples.
This might give the impression that $\optfrac$ is always ``at least as easy'' as $\optint$. But this impression is false: the amortized optimum aggregates the values of $\optint$ over \emph{all} replications of the input, and this aggregation can encode a hard search problem, even when each individual value of $\optint$ is easy to compute. 

An example of this kind, for the classical relaxation in which integrality constraints are dropped, is the 
Motzkin--Straus quadratic program of a graph $G$ \cite{motzkin1965maxima}, whose integer optimum is trivial to compute (it is always $0$), whereas its fractional optimum is a function of the clique number of $G$, which is NP-hard to compute.

We now give an explicit construction, for the amortized relaxation, based on the satisfiability problem.

\begin{proposition}
\label[proposition]{prop:rel-harder-than-opt}
There is a maximization problem $\optint$ that can be computed in polynomial time, whereas computing $\optfrac$ is NP-hard.
\end{proposition}

\begin{proof}
We use two ingredients.

First, a standard fact from combinatorics on words: every finite string $\bx$ has a unique \emph{primitive root} --- a string $\br$ that is not itself a replication of a shorter string --- and a unique exponent $d\geq 1$ such that $\bx = \rep{\br}{d}$; both are computable in time polynomial in $\len(\bx)$, e.g., by checking each divisor of the length of $\bx$. Moreover, for every $k\geq 1$, the primitive root of $\rep{\bx}{k}$ equals the primitive root of $\bx$.

Second, a standard polynomial-time encoding of CNF formulas as binary strings; for a binary string $\br$, we write $\varphi_\br$ for the encoded formula if $\br$ is a valid encoding (in which case we denote by $n_\br$ the number of its Boolean variables), and say that $\varphi_\br$ is undefined otherwise.
For a binary string $\by = (y_1,\ldots,y_n)\in\{0,1\}^n$, we denote by $\operatorname{val}(\by) := \sum_{i=1}^{n} y_i\, 2^{n-i}$ the integer whose $n$-bit binary representation is $\by$, and we interpret $\by$ as a truth assignment to $n$ Boolean variables in the natural way.

We now define the maximization problem in the form \eqref{eq:combinatorial-formulation}: maximize $g(\bx,\by)$ subject to $\by\in\fint(\bx)$.
The input domain is $D := \bigcup\limits_{m\geq 1}\{0,1\}^m$ (all binary vectors). In what follows, $\bx\in D$ is an input with primitive root $\br$ and exponent $d$ (so $\bx = \rep{\br}{d}$), and $\by$ is a candidate solution.
\begin{itemize}
\item The \emph{feasibility set} is
\begin{align*}
\fint(\bx) :=
\begin{cases}
	\{0,1\}^{n_\br} & \text{if $\varphi_\br$ is defined}
	~~~\text{(all truth assignments to the variables of $\varphi_\br$),}
	\\
	\{0\} & \text{otherwise}
	~~~\text{(a single dummy solution).}
\end{cases}
\end{align*}
\item The \emph{objective function} is
\begin{align*}
g(\bx, \by) :=
\begin{cases}
	d & \parbox[t]{9cm}{if $\varphi_\br$ is defined, $\by$ is a satisfying assignment of $\varphi_\br$, and $\operatorname{val}(\by) = d-2$;}
	\\[3mm]
	0 & \text{otherwise.}
\end{cases}
\end{align*}
\item The optimal value is $\optint(\bx) := \max\limits_{\by\in\fint(\bx)} g(\bx,\by)$.
\end{itemize}
In words: a candidate solution $\by$ is a truth assignment to the variables of the formula encoded by the primitive root of the input $\bx$; the solver is rewarded $d$ points if the assignment satisfies the formula \emph{and} its numeric value equals the replication exponent minus $2$; otherwise it is rewarded nothing.

It will be convenient to tuple the value of $\optint$ explicitly. For a primitive string $\br$ and an integer $k\geq 1$, define the indicator
\begin{align*}
s(\br, k) :=
\begin{cases}
1 & \parbox[t]{9cm}{if $\varphi_\br$ is defined, and the binary string with value $k-2$ is a satisfying assignment of $\varphi_\br$;}
\\[3mm]
0 & \text{otherwise.}
\end{cases}
\end{align*}
Since $\operatorname{val}$ is a bijection between $\{0,1\}^{n_\br}$ and $\{0,1,\ldots,2^{n_\br}-1\}$, at most one feasible solution can attain a nonzero objective value --- namely, the $n_\br$-bit binary representation of $d - 2$, when $0\leq d-2\leq 2^{n_\br}-1$. Hence, for every $\bx = \rep{\br}{d}$:
\begin{align}
\label{eq:appendix-optint-value}
\optint(\bx) = \max_{\by\in\fint(\bx)} g(\bx,\by) = d\cdot s(\br, d).
\end{align}

Note that $s(\br,1)=0$ for all $\br$, as no binary string (in our encoding) has value $-1$.
Hence $\optint(\bx)=0$ for any primitive binary string $\bx$.

\emph{$\optint$ is computable in polynomial time.} Given $\bx$ of length $m$: computing the primitive root $\br$ and exponent $d$ takes polynomial time; checking whether $\br$ encodes a CNF formula takes polynomial time; and if it does, then either $d-2 < 0$ or $d - 2 \geq 2^{n_{\br}}$ (in which case $\optint(\bx)=0$, since no feasible $\by$ has $\operatorname{val}(\by)=d-2$), or there is exactly one candidate to check: the algorithm constructs the $n_{\br}$-bit binary representation $\by$ of $d-2$, evaluates $\varphi_{\br}$ on it in polynomial time, and outputs $d$ if $\by$ satisfies $\varphi_\br$ and $0$ otherwise. Note also that $\optint(\bx)\leq d\leq m$, so $\optint(\bx)$ can be represented in space polynomial in $\len(\bx)$.

\emph{Computing $\optfrac$ is NP-hard.} Let $\bx$ be a \emph{primitive} binary string that validly encodes a CNF formula $\varphi_\bx$.%
\footnote{
Primitivity can be guaranteed by the encoding, e.g. by ending every valid encoding with a string of ones with length as the encoding length.
}
For every $k\geq 1$, the primitive root of $\rep{\bx}{k}$ is $\bx$ itself and its exponent is $k$; hence, by \eqref{eq:appendix-optint-value},
\begin{align*}
\frac{\optint(\rep{\bx}{k})}{k} = \frac{k\cdot s(\bx,k)}{k} = s(\bx,k)
\in\{0,1\}.
\end{align*}
Therefore
\begin{align*}
\optfrac(\bx) = \sup_{k\geq 1} s(\bx,k) =
\begin{cases}
1 & \text{if $\varphi_\bx$ is satisfiable (take $k := \operatorname{val}(\by)+2$, where $\by$ is a}\\
& \text{~satisfying assignment of $\varphi_\bx$);}
\\
0 & \text{if $\varphi_\bx$ is unsatisfiable.}
\end{cases}
\end{align*}
So computing $\optfrac(\bx)$ on primitive formula-encoding inputs decides SAT; hence computing $\optfrac$ is NP-hard.

Finally, we verify that $\optfrac$ is finite on all of $D$ (as required of the problems we consider): if $\bx = \rep{\br}{d}$ with $\br$ primitive, then for every $k\geq 1$, the primitive root of $\rep{\bx}{k}$ is $\br$ and its exponent is $dk$, so $\optint(\rep{\bx}{k}) = dk\cdot s(\br,dk) \leq dk$, hence $\optint(\rep{\bx}{k})/k \leq d$ for every $k$, and $\optfrac(\bx)\leq d < \infty$.
\end{proof}

\begin{proposition}
\label[proposition]{cor:p-not-in-ffptas}
The problem $\optint$ of \Cref{prop:rel-harder-than-opt} has no FFPTAS unless ${\sf P}={\sf NP}$. 
\end{proposition}
\begin{proof}
Suppose the problem had an FFPTAS, and run it with $t = 1/2$ on a primitive input $\bx$ encoding a CNF formula $\varphi_\bx$. By \eqref{eq:appendix-optint-value}, $\optint(\bx) = 1\cdot s(\bx,1) = 0$ (since $s(\bx,1)=0$: no assignment $\by$ has $\operatorname{val}(\by) = -1$), while, as computed in the proof of \Cref{prop:rel-harder-than-opt}, $\optfrac(\bx)\in\{0,1\}$ according to the satisfiability of $\varphi_\bx$.
\begin{itemize}
\item If $\varphi_\bx$ is satisfiable, then $(1-t)\optfrac(\bx) = 1/2 > 0 = \optint(\bx)$, so the FFPTAS must return None.
\item If $\varphi_\bx$ is unsatisfiable, then $(1-t)\optfrac(\bx) = 0 = \optint(\bx)$, so the FFPTAS must return the value $0$.
\end{itemize}
Hence the FFPTAS output (None vs.\ a value) decides the satisfiability of $\varphi_\bx$ in polynomial time, which is impossible unless ${\sf P}={\sf NP}$.
\end{proof}

Combining \Cref{prop:rel-harder-than-opt}
and \Cref{cor:p-not-in-ffptas} implies that ${\sf P}\not\subseteq {\sf FFPTAS}$ (unless ${\sf P}={\sf NP}$), and the restriction of \Cref{thm:main} to fractionally-polynomial problems is necessary.

\section{Encoding a minimization problem on graphs}
\label[appendix]{sec:rsp}
The partition problems of \Cref{thm:main} have a particularly convenient input format: the input is a plain list of item sizes, so the replication $\rep{\bx}{k}$ automatically has the intended meaning of ``$k$ independent copies of the instance sharing the same bins''. 
But for other problems, the encoding of the input must be chosen with care. In this section we work out such an example in full: the \emph{Restricted Shortest Path} (RSP) problem.

In this problem, we are given a directed graph in which every edge $e$ has a nonnegative integer \emph{cost} $c_e$ and a nonnegative integer \emph{delay} $d_e$, and a nonnegative integer \emph{delay budget} $B$. The goal is to find a path between two pre-designated nodes $s$ and $t$, of minimum total cost among the paths whose total delay is at most $B$.
RSP is NP-hard (it appears in Garey and Johnson \cite{gareycomputers} as problem ND30, ``Shortest weight-constrained path''). It has an FPTAS, due to Hassin \cite{hassin1992approximation}, later simplified and improved by Lorenz and Raz \cite{lorenz2001simple}.

In the fractional relaxation of the problem, we are allowed to choose fractions of paths, such that the total amount of path fractions from $s$ to $t$ is $1$, the total delay is at most $B$, and the total cost is minimized. This allows smaller total costs. For example, if there are two paths, one with cost $5$ and delay $1$ and one with cost $1$ and delay $5$, and the budget is $3$, the discrete variant must choose the one with cost $5$, whereas the fractional variant can choose $1/2$ of each path, attaining a cost of $3$.
The fractional relaxation can be solved in polynomial time by the ellipsoid method, as well as by a combinatorial algorithm \cite{handler1980dual,mehlhorn2000resource}.


Since Restricted Shortest Path is a minimization problem, we first state the minimization analogues of our definitions.

\subsection{Minimization problems}
\label{sec:rsp-min-defs}

A \emph{minimization problem} is a function $\optint: D\to\Q_{\geq 0}$, where $D = \bigcup_{m\geq 1} X^m$ for some set $X$, exactly as in \Cref{def:problem}. The amortized optimum is defined with an infimum instead of a supremum:
\begin{align}
\label{eq:optfrac-min}
\optfrac(\bx) := \inf_{k\geq 1} \frac{\optint(\rep{\bx}{k})}{k}.
\end{align}
Taking $k=1$ shows that $\optfrac(\bx)\leq \optint(\bx)$ for all $\bx\in D$: as in the maximization case, the amortized optimum is a relaxation of the optimum.
The algorithmic notions are adapted accordingly:
\begin{itemize}
\item An \emph{FPTAS} for a minimization problem $\optint$ is an algorithm that, given $\bx\in D$ and $\epsilon > 0$, runs in time polynomial in $\len(\bx)$ and $1/\epsilon$, and returns a value $v$ such that $\optint(\bx)\leq v\leq (1+\epsilon)\cdot\optint(\bx)$.
\item An \emph{FFPTAS} for a minimization problem $\optint$ is an algorithm that, given $\bx\in D$ and $t > 0$, runs in time polynomial in $\len(\bx)$ and $1/t$, and returns a value $v$ such that $\optint(\bx)\leq v\leq (1+t)\cdot\optfrac(\bx)$ (if such a value exists, that is, if $\optint(\bx)\leq (1+t)\cdot\optfrac(\bx)$), or None (otherwise).
\end{itemize}
As in \Cref{rem:feasible-value}, the requirement $v\geq\optint(\bx)$ plays the role of ``$v$ is the value of a feasible solution'': for a minimization problem, the value of any feasible solution is at least the optimal value.
A minimization problem is \emph{fractionally-polynomial} under the same two conditions as before: $\optfrac(\bx)$ is computable in time polynomial in $\len(\bx)$, and $\optint(\bx)$ is representable in space polynomial in $\len(\bx)$.

\subsection{Input encoding}
\label{sec:rsp-encoding}
The base set $X$ contains tuples of two different kinds:
\begin{itemize}
\item Edge tuples $(Edge: u, v, c, d)$, denoting a directed edge from vertex $u$ to vertex $v$, with cost $c$ and delay $d$. The multiset of edge tuples defines a directed multigraph $G(\bx)$, whose vertex set is the set of integers appearing in the $u,v$ fields of the tuples of $\bx$.
\item Demand tuples $(Demand: b)$, denoting a request to route a path from the source $s$ to the target $t$, with a budget of $b$. We denote by $r(\bx)$ the number of demand tuples in $\bx$, and by $B(\bx)$ the sum of $b$ in all demand tuples.
\end{itemize}
A \emph{path} always means a simple directed path from $s$ to $t$ in $G(\bx)$.
We denote by ${\cal P}$ the (finite) set of paths, and for a path $P$ we write $c(P)$ and $d(P)$ for its total cost and total delay. 
A \emph{solution} for $\bx$ is a selection of $r$ paths $P_1,\ldots,P_r\in{\cal P}$, not necessarily distinct (one path per demand); it is \emph{feasible} if $\sum_{j=1}^r d(P_j)\leq B(\bx)$. We define:
\begin{align}
\label{eq:rsp-optint}
\optint(\bx) :=
\begin{cases}
\displaystyle
\min\Big\{\sum_{j=1}^{r} c(P_j) ~:~ (P_1,\ldots,P_r) \text{ is a feasible solution for } \bx\Big\}
\\
~~~~~~~~~~ \text{if a feasible solution exists;}
\\
0 ~~~~~~~~ \text{otherwise.}
\end{cases}
\end{align}
We call $\bx$ \emph{feasible} in the first case and \emph{infeasible} in the second (an input with no demand tuples is feasible, with $\optint(\bx)=0$, by the convention that an empty sum is $0$). A standard RSP instance with budget $B$ corresponds to a feasible input with a single demand tuple carrying $b = B$; the convention $\optint := 0$ on infeasible inputs merely makes $\optint$ total on $D$, and such inputs are recognizable in polynomial time (\Cref{lem:rsp-feasibility} below).

The key point of this encoding is the semantics of replication. The input $\rep{\bx}{k}$ decodes to: the multigraph $G(\bx)$ with every edge duplicated $k$ times (parallel edges with identical costs and delays, which are clearly immaterial for path selection), $kr$ demands, and the pooled budget $k\cdot B(\bx)$. Hence:
\begin{align}
\label{eq:rsp-replication}
\optint(\rep{\bx}{k}) = \min\Big\{ \sum_{j=1}^{k r(\bx)} c(P_{j}) ~:~ P_1,\ldots,P_{k r(\bx)}\in{\cal P}, \text{ and } \sum_{j=1}^{k r(\bx)} d(P_{j}) \leq k\cdot B(\bx)\Big\},
\end{align}
when a selection satisfying the constraint exists.
The paths chosen for the different demands may differ; it is exactly this freedom, combined with the pooled budget, that makes the amortized optimum a nontrivial relaxation.

\begin{lemma}
\label{lem:rsp-feasibility}
For $\bx\in D$, let $\mu$ denote the minimum delay of a path in ${\cal P}$ (and $\mu := \infty$ if ${\cal P}$ is empty). Then:

(a) $\bx$ is feasible if and only if $r(\bx) = 0$ or $\mu\cdot r(\bx) \leq B(\bx)$; 

(b) This condition can be checked in polynomial time; and  ---

(c) $\bx$ is feasible if and only if $\rep{\bx}{k}$ is feasible, for every $k\geq 1$.
\end{lemma}
\begin{proof}
(a) If $r(\bx)=0$ then $\bx$ is feasible by definition; if $\mu\cdot r(\bx) \leq B(\bx)$ then selecting the minimum-delay path $r$ times is a feasible solution; if both these conditions do not  hold, then there is no solution satisfying the budget constraints.

(b) The value $\mu$ is computable in polynomial time by a shortest-path computation with respect to the (nonnegative) delays. 

(c) By \eqref{eq:rsp-replication}, the minimum possible total delay in $\rep{\bx}{k}$ is $\mu \cdot k r(\bx)$, and the budget is $k B(\bx)$; so feasibility of $\rep{\bx}{k}$ is the condition $\mu \cdot k r(\bx) \leq k B(\bx)$, which is equivalent to feasibility of $\bx$.
\end{proof}

\subsection{The amortized optimum equals the fractional path relaxation}
\label{sec:rsp-amortized}

The classical fractional relaxation of RSP
can be presented by the following linear program:
\begin{align}
\label{eq:rsp-lp}
\operatorname{OptLP}(\bx) := \min & \sum_{P\in{\cal P}} \lambda_P\cdot  c(P)
\\
\notag
\text{subject to} 
&
\sum_{P\in{\cal P}} \lambda_P = r(\bx),
\\
\notag
&
\sum_{P\in{\cal P}} \lambda_P\cdot  d(P) \leq B(\bx),
\\
\notag
&
\lambda_P \geq 0 ~~~~~ \forall P\in{\cal P},
\end{align}
with the convention $\operatorname{OptLP}(\bx) := 0$ if the program is infeasible or $r(\bx) = 0$. 
Note that the program \eqref{eq:rsp-lp} is feasible exactly when $\bx$ is feasible: a fractional routing minimizing the total delay puts all its mass on minimum-delay paths, so feasibility of \eqref{eq:rsp-lp} is again the condition $\mu r(\bx)\leq B(\bx)$ of \Cref{lem:rsp-feasibility}. Since ${\cal P}$ is finite, the feasible region of \eqref{eq:rsp-lp} is a (compact) polytope, and the minimum is attained.

\begin{theorem}
\label{thm:rsp-optfrac}
For every $\bx\in D$:~ $\optfrac(\bx) = \operatorname{OptLP}(\bx)$. Moreover, the infimum in \eqref{eq:optfrac-min} is attained.
\end{theorem}
\begin{proof}
If $\bx$ is infeasible, then by \Cref{lem:rsp-feasibility} every $\rep{\bx}{k}$ is infeasible, so $\optint(\rep{\bx}{k}) = 0$ for all $k$ and $\optfrac(\bx) = 0 = \operatorname{OptLP}(\bx)$ (both by convention; the infimum is attained at $k=1$). The same holds if $r(\bx)=0$.

Assume from now on that $\bx$ is feasible and $r(\bx)\geq 1$; then all $\rep{\bx}{k}$ are feasible as well, and both quantities are given by the minima in \eqref{eq:rsp-replication} and \eqref{eq:rsp-lp}.
(Formally, a path in $G(\rep{\bx}{k})$ may use different parallel copies of the same edge of $G(\bx)$; since the copies have identical costs and delays, we may identify every such path with the corresponding path of $G(\bx)$, and we do so throughout.)

\emph{Lower bound ($\operatorname{OptLP}(\bx) \leq \optint(\rep{\bx}{k})/k$ for every $k\geq 1$).}
Let $(P_{1},\ldots,P_{kr})$ be an optimal solution of $\rep{\bx}{k}$, as in \eqref{eq:rsp-replication}. Define, for each path $P\in{\cal P}$:
\begin{align*}
\lambda_P := \frac{|\{j : P_{j} = P\}|}{k}.
\end{align*}
Then $\sum_P \lambda_P = kr/k = r$, and
\begin{align*}
\sum_{P} \lambda_P\, d(P) = \frac{1}{k}\sum_{j=1}^{kr} d(P_{j}) \leq \frac{1}{k}\cdot k B(\bx) = B(\bx),
\end{align*}
and $\lambda_P\geq 0$,
so $\lambda$ is feasible for \eqref{eq:rsp-lp}. Its objective value is
$\sum_{P} \lambda_P\, c(P) = \frac{1}{k}\sum_{j} c(P_{j}) = \optint(\rep{\bx}{k})/k$.
Hence $\operatorname{OptLP}(\bx)\leq \optint(\rep{\bx}{k})/k$. Taking the infimum over $k$ gives $\operatorname{OptLP}(\bx)\leq\optfrac(\bx)$.

\emph{Upper bound ($\optint(\rep{\bx}{q}) \leq q\cdot\operatorname{OptLP}(\bx)$ for some integer $q\geq 1$ of length polynomial in $\len(\bx)$).}
The program \eqref{eq:rsp-lp} is a linear program in the variables $(\lambda_P)_{P\in{\cal P}}$, with two constraints other than the nonnegativity constraints. 
Hence, it has an optimal solution $\lambda^*$ in which at most two variables are nonzero 
(see e.g. \cite{matousek2007understanding}).

\begin{itemize}
\item If $\lambda^*$ has a single nonzero variable, say on path $P_1$, then $\lambda^*_{P_1} = r(\bx)$ by the first constraint of \eqref{eq:rsp-lp}, which is an integer. Set $q := 1$.
\item If $\lambda^*$ has exactly two nonzero coordinates, say on paths $P_1$ and $P_2$, then both constraints of \eqref{eq:rsp-lp} are tight at $\lambda^*$ (as it is a corner of the feasible region), so $(\lambda^*_{P_1}, \lambda^*_{P_2})$ is the unique solution of the system
\begin{align*}
\lambda_{P_1} + \lambda_{P_2} = r(\bx) && d({P_1})\, \lambda_{P_1} + d({P_2})\, \lambda_{P_2} = B(\bx),
\end{align*}
whose coefficients and right-hand sides are integers.
{Moreover, we can assume that $d({P_1}) \neq  d({P_2})$, otherwise we could move all mass to a path with the smaller cost and land at the previous case of a single non-zero variable.} 
By Cramer's rule, $\lambda^*_{P_1}$ and $\lambda^*_{P_2}$ are rational numbers with common denominator $q := |d({P_1}) - d({P_2})|$, which satisfies $1\leq q \leq \Delta$, where $\Delta$ denotes the sum of all edge delays in $\bx$ (see, e.g., \cite{schrijver1998theory} for these standard facts). In particular, $\len(q)\leq\len(\bx)$.
\end{itemize}
In both cases, $q\, \lambda^*_P$ is a nonnegative integer for every $P\in{\cal P}$, and $\sum_P q\, \lambda^*_P = q r(\bx)$.

Construct a solution for $\rep{\bx}{q}$ where each $P\in{\cal P}$ appears $q\, \lambda^*_P$ times. 
The total delay is
$\sum_P q\, \lambda^*_P\, d(P) \leq q B(\bx)$, so the solution is feasible for $\rep{\bx}{q}$; and its total cost is
$\sum_P q\, \lambda^*_P\, c(P) = q\cdot\operatorname{OptLP}(\bx)$. Hence $\optint(\rep{\bx}{q}) \leq q\cdot\operatorname{OptLP}(\bx)$, so $\optfrac(\bx)\leq \operatorname{OptLP}(\bx)$.

Combining the two bounds, $\optfrac(\bx) = \operatorname{OptLP}(\bx)$; and since by the lower bound every term $\optint(\rep{\bx}{k})/k$ is at least $\operatorname{OptLP}(\bx)$, while by the upper bound the term at $k=q$ is at most (hence exactly) $\operatorname{OptLP}(\bx)$, the infimum is attained at $k=q$.
\end{proof}

The results in \cite{mehlhorn2000resource} imply that OptLP can be solved in polynomial time when $r(\bx)=1$. The following simple lemma shows a reduction from arbitrary $r(\bx)$ to $r=1$. 
\begin{lemma}
\label{lem:eersp-layering}
Let $\bx\in D$ be feasible with $r(\bx)\geq 1$ demands. Construct the graph $G^{(r(\bx))}$ consisting of $r(\bx)$ copies of $G(\bx)$ in series: for each $j\in\{1,\ldots,r(\bx)-1\}$, add an edge with cost $0$ and delay $0$ from the $t$-vertex of copy $j$ to the $s$-vertex of copy $j+1$. Then

(a)  $\optint(\bx)$ equals the minimum cost of a path in $G^{(r(\bx))}$, from the $s$-vertex of copy $1$ to the $t$-vertex of copy $r(\bx)$, whose delay is at most $B(\bx)$.

(b) Similarly, $\optfrac(\bx)$ equals the minimum cost of a fractional path in $G^{(r(\bx))}$ whose delay is at most $B(\bx)$.
\end{lemma}
\begin{proof}
(a) A selection of $r(\bx)$ paths in $G(\bx)$ corresponds to a walk through $G^{(r(\bx))}$ that crosses the copies in order, with the same total cost and delay; conversely, any walk in $G^{(r(\bx))}$ between the two designated vertices crosses each copy in a walk from its $s$-vertex to its $t$-vertex, which 
yields a selection of $r(\bx)$  paths with the same total cost and delay. 
Hence the two minima coincide.

(b) Let $(\lambda_P)_{P\in\mathcal{P}}$ be a distribution of total mass $r(\bx)$ among paths in $G(\bx)$. We can construct from it a distribution of total mass $1$ among paths in $G^{(r(\bx))}$, with the same total cost and delay, using a simple greedy procedure. First, replace any paths $P$ with $\lambda_P>1$ with two or more paths identical to $P$, each of which with $\lambda\leq 1$. 
Next, choose a path $P_1$ with smallest $\lambda_{P_1}$, and assign it to the first copy of $G$ in $G^{(r(\bx))}$; choose another path $P_2$, and assign a $\lambda_{P_1}$ fraction of it to the second copy of $G$; proceed this way until there is a path from the source to the target of $G^{(r(\bx))}$, with weight $\lambda_{P_1}$. Remove $P_1$, as well as a fraction $\lambda_{P_1}$ from all other participating paths, and repeat.

Conversely, given a distribution of total mass $1$ among paths in $G^{(r(\bx))}$, we can decompose each path with weight $\lambda_P$ into $r(\bx)$ different paths with the same weight $\lambda_P$. This results in a distribution of total mass $r(\bx)$ among paths in $G(\bx)$, with the same total cost and delay.
Hence the two minima coincide.
\end{proof}

\begin{corollary}
\label{cor:optfrac-polytime}
$\optfrac$ can be solved in polynomial time.
\end{corollary}
\begin{proof}
Use the reduction of \Cref{lem:eersp-layering} and the algorithm of \cite{mehlhorn2000resource}.
\end{proof}
This shows that the RSP problem (in our encoding) is fractionally-polynomial.
As it is NP-hard and has an FPTAS, a natural open question arises:
\begin{open}
Does the Restricted Shortest Path problem, under the pooled-budget encoding of \Cref{sec:rsp-encoding}, have an FFPTAS?
\end{open}

\subsection{Expensive-edge RSP: an NP-hard special case with an FFPTAS}
\label{sec:eersp}
Towards addressing the open problem, we show here a special case of RSP, which is still NP-hard, and has an FFPTAS. This is a minimization-side analogue of Part (2) of \Cref{thm:main}.
We call this special case \emph{Expensive-edge RSP} (EERSP) --- RSP on graphs in which every edge leaving the source is very expensive (so every path begins by paying a large fixed cost $C_0$).


For an input $\bx$, we denote by $\bar{E}(\bx)$ the \emph{set} of distinct edge tuples of $\bx$ (that is, the edge tuples after removing duplicates).

\begin{definition}[EERSP]
\label{def:eersp}
An input $\bx\in D$ belongs to the family EERSP if it is feasible and satisfies the following two conditions:
\begin{enumerate}
\item There is an integer $C_0\geq 1$ such that every edge tuple of $\bx$ whose tail is $s$ has cost $C_0$ and delay $0$, and at least one such tuple exists. We call these the \emph{offset tuples}, and the remaining tuples of $\bar{E}(\bx)$ the \emph{core tuples}.

We denote by $\bar{S}(\bx)$ the sum of the costs of the core tuples, and by $\bar{\Delta}(\bx)$ the sum of the delays of all tuples of $\bar{E}(\bx)$.
\item $C_0 \geq \max(\bar{\Delta}(\bx), 1)\cdot \bar{S}(\bx)$.
\end{enumerate}
\end{definition}
Since every path starts at $s$, and a simple path leaves $s$ exactly once, condition 1 says that every path must begin with an offset tuple, and hence contains exactly one edge of cost $C_0$ and delay $0$ leaving $s$. Condition 2 quantifies how expensive this edge must be, in terms of two quantities that depend only on the set of \emph{distinct} tuples. 
The condition uses $\max(\bar{\Delta}(\bx), 1)$ rather than just $\bar{\Delta}(\bx)$, as $\bar{\Delta}(\bx)$ might be $0$. This restriction does not lose much generality, as when $\bar{\Delta}(\bx)=0$ RSP reduces to standard shortest-path computation, which is in ${\sf P}$ and hence has an FFPTAS.

This makes the family replication-closed:

\begin{lemma}
\label{lem:eersp-closed}
For every $\bx\in D$ and every integer $k\geq 1$:~ $\bx\in$ EERSP if and only if $\rep{\bx}{k}\in$ EERSP. Moreover, membership in EERSP can be decided in polynomial time.
\end{lemma}
\begin{proof}
Replication duplicates every tuple, so $\bar{E}(\rep{\bx}{k}) = \bar{E}(\bx)$; hence conditions 1 and 2, as well as the quantities $\bar{S}$ and $\bar{\Delta}$, are invariant under replication. Feasibility is invariant under replication by \Cref{lem:rsp-feasibility}. Membership is decided by scanning the tuples (for conditions 1--2) and by the feasibility test of \Cref{lem:rsp-feasibility}.
\end{proof}

\begin{lemma}
\label{lem:eersp-path-bounds}
Let $\bx\in$ EERSP and $k\geq 1$. Then every path $P$ in $G(\rep{\bx}{k})$ satisfies
\begin{align*}
C_0 ~\leq~ c(P) ~\leq~ C_0 + \bar{S}(\bx),
\qquad\qquad
d(P) ~\leq~ \bar{\Delta}(\bx).
\end{align*}
\end{lemma}
\begin{proof}
A path is simple, so it visits every vertex at most once; hence it traverses at most one edge between every ordered pair of vertices, and in particular it never traverses two copies of the same tuple. Therefore, the tuples traversed by $P$ form a subset of $\bar{E}(\bx)$ (recall that $\bar{E}(\rep{\bx}{k}) = \bar{E}(\bx)$). Since $P$ starts at $s$ and leaves it exactly once, exactly one tuple of this subset has tail $s$, and by condition 1 of \Cref{def:eersp} it has cost $C_0$ and delay $0$; every other tuple of the subset is a core tuple. The claimed bounds follow by summing costs and delays over the subset: the cost is $C_0$ plus a sum of core costs (between $C_0$ and $C_0 + \bar{S}(\bx)$), and the delay is a sum of core delays (at most $\bar{\Delta}(\bx)$).
\end{proof}

We will use one more lemma, that shows how to round the fractional optimum to an integral solution with a small \emph{additive} loss.

\begin{lemma}
\label{lem:eersp-rounding}
Let $\bx\in$ EERSP with $r(\bx)\geq 1$. Then
\begin{align*}
\optint(\bx) ~\leq~ \optfrac(\bx) + r(\bx)\cdot\bar{S}(\bx),
\end{align*}
and a feasible solution of $\bx$ with total cost at most $\optfrac(\bx) + r(\bx)\cdot\bar{S}(\bx)$ can be computed in polynomial time.
\end{lemma}
\begin{proof}
Consider the layered graph $G^{(r(\bx))}$ constructed in the proof of \Cref{lem:eersp-layering}, with a single demand and budget $B(\bx)$; as in the proof of \Cref{cor:optfrac-polytime}, the optimal value of its fractional relaxation is $\optfrac(\bx)$.
The algorithm of \cite{mehlhorn2000resource} computes, in polynomial time, this optimal value together with an optimal fractional solution supported on at most two paths of $G^{(r(\bx))}$ --- say $Q_1$ and $Q_2$, with weights $\alpha$ and $1-\alpha$ for some $\alpha\in[0,1]$
(the two paths are the endpoints of a segment of the lower convex hull of the delay--cost profiles of the paths; this two-path structure of the fractional optimum goes back to \cite{handler1980dual}).
Assume without loss of generality that $d(Q_1)\leq d(Q_2)$, and let $(P_1,\ldots,P_{r(\bx)})$ be the solution of $\bx$ that the path $Q_1$ decodes to (\Cref{lem:eersp-layering}); we return this solution. It is computed in polynomial time; it remains to bound its delay and cost.

\emph{Delay.} $d(Q_1) \leq \alpha\, d(Q_1) + (1-\alpha)\, d(Q_2) \leq B(\bx)$, by the choice of $Q_1$ as the lower-delay path and the feasibility of the fractional solution. Hence the returned solution is feasible.

\emph{Cost.} Since $\optfrac(\bx) = \alpha\, c(Q_1) + (1-\alpha)\, c(Q_2)$,
\begin{align*}
c(Q_1) ~=~ \optfrac(\bx) + (1-\alpha)\big(c(Q_1) - c(Q_2)\big)
~\leq~ \optfrac(\bx) + \big|c(Q_1) - c(Q_2)\big|.
\end{align*}
Now, each of $Q_1, Q_2$ decodes to $r(\bx)$ paths of $G(\bx)$, and the connector edges of $G^{(r(\bx))}$ have cost $0$; so by \Cref{lem:eersp-path-bounds}, applied to each of the $r(\bx)$ paths,
\begin{align*}
r(\bx)\cdot C_0 ~\leq~ c(Q_i) ~\leq~ r(\bx)\cdot\big(C_0 + \bar{S}(\bx)\big)
&& \text{for } i\in\{1,2\}.
\end{align*}
Both costs lie in an interval of length $r(\bx)\cdot\bar{S}(\bx)$, so $|c(Q_1) - c(Q_2)| \leq r(\bx)\cdot\bar{S}(\bx)$, and the cost of the returned integral solution is at most $\optfrac(\bx) + r(\bx)\cdot\bar{S}(\bx)$.
\end{proof}

Note where condition 1 of \Cref{def:eersp} entered: the rounding loss is bounded by the \emph{difference} of the two path costs, and since every path carries the common offset cost $C_0$ exactly once per demand, the offset cancels in the difference, which is therefore bounded by the core costs alone. The same condition also bounds the fractional optimum from below:

\begin{lemma}
\label{lem:eersp-scale}
For every $\bx\in$ EERSP:~ $\optfrac(\bx) \geq r(\bx)\cdot C_0$.
\end{lemma}
\begin{proof}
By \Cref{lem:eersp-path-bounds}, $c(P)\geq C_0$ for every path $P$. Hence the objective of the path LP \eqref{eq:rsp-lp} satisfies $\sum_P \lambda_P\cdot c(P) \geq C_0 \sum_P \lambda_P = C_0\cdot r(\bx)$ for every feasible $\lambda$, and by \Cref{thm:rsp-optfrac}, $\optfrac(\bx) = \operatorname{OptLP}(\bx) \geq r(\bx)\cdot C_0$.
\end{proof}

We can now present the FFPTAS. Its two cases are separated by comparing $t$ to the ratio $\bar{S}(\bx)/C_0$. By \Cref{lem:eersp-rounding,lem:eersp-scale}, this ratio is an upper bound on the relative gap:
\begin{align*}
\frac{\optint(\bx) - \optfrac(\bx)}{\optfrac(\bx)}
~\leq~ \frac{r(\bx)\cdot\bar{S}(\bx)}{r(\bx)\cdot C_0}
~=~ \frac{\bar{S}(\bx)}{C_0}
~\leq~ \frac{1}{\max(\bar{\Delta}(\bx),1)},
\end{align*}
where the last inequality is condition 2 of \Cref{def:eersp}. So when $t$ is above the ratio, the required value is guaranteed to exist, though computing the corresponding solution is non-trivial and requires the rounding algorithm of \Cref{lem:eersp-rounding}. 

\begin{algorithm}[!h] \caption{\qquad FFPTAS$[\text{EERSP}](\bx,t)$} \label{alg:eersp-ffptas}
\begin{algorithmic}[1]
\State If $\bx\notin$ EERSP or $r(\bx)=0$: \Return $0$.
\State Compute $\optfrac(\bx)$, using \Cref{cor:optfrac-polytime}.
\If {\label{line:eersp-round} $t \geq \bar{S}(\bx)\,/\,C_0$}
\State Compute the rounded solution of \Cref{lem:eersp-rounding}; \Return its cost $v$.
\Else {~($t < \bar{S}(\bx)\,/\,C_0$)}
\State \label{line:eersp-dp} Compute $\optint(\bx)$ exactly, by the classical dynamic program for RSP on the layered graph $G^{(r(\bx))}$ of \Cref{lem:eersp-layering}, over the delay values $0,1,\ldots,\min\big(B(\bx),\, r(\bx)\cdot\bar{\Delta}(\bx)\big)$.
\State If $\optint(\bx) \leq (1+t)\cdot\optfrac(\bx)$: \Return $\optint(\bx)$; otherwise \Return None.
\EndIf
\end{algorithmic}
\end{algorithm}

\begin{theorem}
\label{thm:eersp-ffptas}
\Cref{alg:eersp-ffptas} is an FFPTAS for RSP restricted to EERSP inputs: given $\bx\in$ EERSP and $t>0$, it runs in time polynomial in $\len(\bx)$ and $1/t$, and returns a value $v$ with $\optint(\bx)\leq v\leq (1+t)\optfrac(\bx)$ if such a value exists, or None otherwise.
\end{theorem}
\begin{proof}
Assume $\bx\in$ EERSP with $r(\bx)\geq 1$ (for $r(\bx)=0$, the output $0$ is correct, as $\optint(\bx)=\optfrac(\bx)=0$).

\underline{Case $t \geq \bar{S}(\bx)/C_0$ (Line \ref{line:eersp-round}).} By \Cref{lem:eersp-scale}, $t\cdot\optfrac(\bx) \geq t\cdot r(\bx)\, C_0 \geq r(\bx)\cdot\bar{S}(\bx)$. The value $v$ returned is the cost of an actual feasible solution of $\bx$, so $v\geq\optint(\bx)$; and by \Cref{lem:eersp-rounding},
\begin{align*}
v ~\leq~ \optfrac(\bx) + r(\bx)\cdot\bar{S}(\bx) ~\leq~ \optfrac(\bx) + t\cdot\optfrac(\bx) ~=~ (1+t)\,\optfrac(\bx).
\end{align*}
Hence $v$ satisfies both FFPTAS requirements.
All the computations in this case take time polynomial in $\len(\bx)$.

\underline{Case $t < \bar{S}(\bx)/C_0$ (Line \ref{line:eersp-dp}).} By \Cref{lem:eersp-path-bounds}, every solution of $\bx$ has total delay at most $r(\bx)\cdot\bar{\Delta}(\bx)$, so restricting the dynamic program to delay values up to $\min(B(\bx),\, r(\bx)\cdot\bar{\Delta}(\bx))$ does not change its result, and by \Cref{lem:eersp-layering} it computes $\optint(\bx)$ exactly. For the run-time, note that in this case $\bar{S}(\bx)\geq 1$, so by condition 2 of \Cref{def:eersp},
\begin{align*}
\frac{1}{t} ~>~ \frac{C_0}{\bar{S}(\bx)} ~\geq~ \max(\bar{\Delta}(\bx),1) ~\geq~ \bar{\Delta}(\bx).
\end{align*}
In this case, $\optint(\bx)$ can be computed exactly using a standard dynamic programming algorithm, that computes for each possible delay value the smallest cost for that delay. The number of total delay values in the layered graph $G^{(r(\bx))}$ is at most $r(\bx)\bar{\Delta}(\bx) +1 \leq r(\bx)/t+1$, and $r(\bx)\leq \len(\bx)$, so the algorithm runs in time polynomial in $\len(\bx)$ and $1/t$. 
Given the exact values $\optint(\bx)$ and $\optfrac(\bx)$, the final comparison implements the FFPTAS requirement literally: it returns the value $\optint(\bx)$ --- which trivially satisfies $\optint(\bx)\leq v$ and, by the comparison, $v\leq(1+t)\optfrac(\bx)$ --- exactly when such a value exists.
\end{proof}

At the same time, the expensive edge does not make the problem any easier to solve exactly:

\begin{proposition}
\label{prop:eersp-hard}
Computing $\optint(\bx)$ on EERSP inputs is NP-hard.
\end{proposition}
\begin{proof}
By reduction from {\sc Partition} \cite{gareycomputers}. Given positive integers $a_1,\ldots,a_m$ with sum $2A$, construct the standard chain instance of RSP (the one used in the NP-hardness proof of RSP, problem ND30 of \cite{gareycomputers}): a path of vertices $u = u_0, u_1, \ldots, u_m = t$, where between $u_{i-1}$ and $u_i$ there are two parallel edge tuples, one with cost $a_i$ and delay $0$, and one with cost $0$ and delay $a_i$. A path from $u$ to $t$ chooses, for each $i$, whether to ``pay'' $a_i$ in cost or in delay; so there is a $u$--$t$ path with delay at most $A$ and cost at most $A$ if and only if $(a_1,\ldots,a_m)$ has a partition into two sets of sum $A$ each.
Note that in this graph $S(\bx) = \Delta(\bx)=2A$.

Now prepend an edge $s\to u$ with cost $C_0$ and delay $0$, with $C_0 := \Delta(\bx)\cdot S(\bx) = 2A \cdot 2A$, and add one demand tuple with budget $A$.
The instance satisfies both conditions in the EERSP definition, its length is polynomial in the length of the {\sc Partition} instance, and
\begin{align*}
\optint(\bx) \leq C_0 + A
\quad\iff\quad
(a_1,\ldots,a_m) \text{ has a partition into two sets of sum } A,
\end{align*}
since every $s$--$t$ path consists of the offset tuple followed by a $u$--$t$ path. So a polynomial-time algorithm for $\optint$ on EERSP inputs would solve {\sc Partition} in polynomial time.
\end{proof}

Finally, since EERSP is replication-closed (\Cref{lem:eersp-closed}), it yields a total minimization problem in the sense of \Cref{def:problem}, to which all the notions of this paper apply directly. Define, on the domain $D$ of \Cref{sec:rsp-encoding}:
\begin{align*}
\optint_{E}(\bx) :=
\begin{cases}
\optint(\bx) & \text{if } \bx\in \text{EERSP},
\\
0 & \text{otherwise.}
\end{cases}
\end{align*}

\begin{corollary}
\label{cor:eersp-part2}
$\optint_{E}$ is a fractionally-polynomial minimization problem that has an FFPTAS, but has no polynomial-time algorithm unless ${\sf P} = {\sf NP}$.
\end{corollary}
\begin{proof}
By \Cref{lem:eersp-closed}, $\bx\in$ EERSP if and only if $\rep{\bx}{k}\in$ EERSP; hence for $\bx\in$ EERSP, all the terms of \eqref{eq:optfrac-min} for $\optint_E$ coincide with those for $\optint$, so $\optfrac_E(\bx) = \optfrac(\bx)$; and for $\bx\notin$ EERSP, all the terms are $0$, so $\optfrac_E(\bx) = 0$. Fractional-polynomiality follows: decide membership in polynomial time (\Cref{lem:eersp-closed}), then either output $0$ or compute $\optfrac(\bx)$ by \Cref{cor:optfrac-polytime}; and $\optint_E(\bx)\leq\optint(\bx)$ is representable in polynomial space. \Cref{alg:eersp-ffptas} (whose first line implements the membership check) is an FFPTAS for $\optint_E$ by \Cref{thm:eersp-ffptas}. Finally, a polynomial-time algorithm for $\optint_E$ would compute $\optint$ on EERSP inputs, which is impossible unless ${\sf P}={\sf NP}$ by \Cref{prop:eersp-hard}.
\end{proof}

\Cref{cor:eersp-part2} is a minimization-side analogue of Part (2) of \Cref{thm:main}, established on a graph problem.

\end{document}